\documentclass[12pt]{article}
\usepackage[a4paper, total={6in, 8in}]{geometry}
\usepackage{amsmath}
\usepackage{amsfonts}
\usepackage{amsthm}
\usepackage{hyperref}
\usepackage{bbm}
\usepackage{wrapfig}
\usepackage{scrextend}
\usepackage{amssymb}
\usepackage{graphicx}
\usepackage{tabularray}
\graphicspath{ {./images/} }
\newcommand*{\p}{\mathbb{P}}
\usepackage{xcolor}
\usepackage{dsfont}
\usepackage{natbib}
\usepackage{graphicx}
\graphicspath{ {./images/} }
\usepackage{float}
 \usepackage{algorithm} 

\UseTblrLibrary{booktabs}

\newcommand{\NB}{\textcolor{red}}

\def\lta{\mathrel{\hbox{\rlap{\hbox{\lower4pt\hbox{$\sim$}}}\hbox{$<$}}}}

\newcommand{\norm}[1]{\left\lVert#1\right\rVert}
\usepackage{setspace}

\newtheorem{Lemma}{Lemma}[section]

\newtheorem{theorem}[Lemma]{Theorem}
\newtheorem{corollary}[Lemma]{Corollary}

\newtheorem{assumption}{Assumption}[section]

\usepackage[utf8]{inputenc}

\theoremstyle{definition}

\date{}

\definecolor{darkblue}{rgb}{.1, 0.1,.8}
\definecolor{darkgreen}{rgb}{0,0.8,0.2}
\definecolor{darkred}{rgb}{.8, .1,.1}

\usepackage{mathtools}
\mathtoolsset{showonlyrefs}
\newcommand*{\E}{\mathbb{E}}

\newcommand*{\R}{\mathbb{R}}
\renewcommand{\P }{{\mathbb P}}
\newcommand{\1}{\mathbbm{1}}

\begin{document}

\title{\Large \bf Detecting relevant dependencies under measurement error with applications to the analysis of planetary system evolution}

\author{
{\rm Patrick Bastian}\\
Ruhr-University Bochum
\and
{\rm Nicolai Bissantz}\\
Ruhr-University Bochum
}

\maketitle

\begin{abstract}

Exoplanets play an important role in understanding the mechanics of planetary system formation and orbital evolution. In this context the correlations of different parameters of the planets and their host star are useful guides in the search for explanatory mechanisms. Based on a reanalysis of the data set from \cite{figueria14} we study the as of now still poorly understood correlation between planetary surface gravity and stellar activity of Hot Jupiters. Unfortunately, data collection often suffers from measurement errors due to complicated and indirect measurement setups, rendering standard inference techniques unreliable. \\ 

We present new methods to estimate and test for correlations in a deconvolution framework and thereby improve the state of the art analysis of the data in two directions. First, we are now able to account for additive measurement errors which facilitates reliable inference. Second we test for relevant changes, i.e. we are testing for correlations exceeding a certain threshold $\Delta$. This reflects the fact that small nonzero correlations are to be expected for real life data almost always and that standard statistical tests will therefore always reject the null of no correlation given sufficient data. Our theory focuses on quantities that can be estimated by U-Statistics which contain a variety of correlation measures. We propose a bootstrap test and establish its theoretical validity. As a by product we also obtain confidence intervals. Applying our methods to the Hot Jupiter data set from \cite{figueria14}, we observe that taking into account the measurement errors yields smaller point estimates and the null of no relevant correlation is rejected only for very small $\Delta$. This demonstrates the importance of considering the impact of measurement errors to avoid misleading conclusions from the resulting statistical analysis.
\end{abstract}

\section{Introduction}
\label{introduction}

Deconvolution-like problems are commonplace in a diverse range of areas and methods such as accounting for measurement errors in econometrics \cite{Kato2021}  or  signal de-blurring in image analysis \cite{Qiu2005}. They are a particular class of inverse problems and a common statistic model is given by the additive noise model
\begin{align}
\label{IntroModel}
    Z_i=X_i+\epsilon_i \quad i=1,...,n
\end{align}
where only the $Z_i \in \R^p$ are observed and one is interested in the density $f$ of $X_i$. Equation \eqref{IntroModel} then implies the relationship
\begin{align}
    g=f \star \psi
\end{align}
where $g$ is the density of $Z_i$ and $\psi$ the density of $\epsilon_i$ which is usually assumed to be known. There is by now a vast literature concerned with estimation of $f$ in this setup (see e.g.  \cite{fan91a}, \cite{fan91b}, \cite{dh93}, \cite{bissantz}). We contribute by extending estimation and inference of $U$-statistics to the model \eqref{IntroModel}.\\

\textbf{This research is motivated by a problem in understanding the formation and evolution of planetary systems.} Hot Jupiters are an only fairly recently discovered class of stellar objects that have been observed for the first time only a scant few decades ago. They are gas giants with mass comparable to or larger than that of Jupiter and extremely short orbital periods lasting only a few days as they typically orbit their parent star at rather short distances. While the possibility of the existence of such planets had already been considered by \cite{struve52} they have not been predicted by planet system formation models, thereby pointing to some open problems in this area. We refer to \cite{anrev18, JGR21} for more details. Understanding the physical characteristics of Hot Jupiters and the role they play in the evolution of planetary systems is therefore an attractive avenue towards closing gaps in planet system formation models. A first step in the search for explanatory mechanisms is to analyze the correlations between different physical quantities pertinent to this situation. Regarding hot Jupiters, a range of potentially important physical characteristics of the planet and its star have been analysed for correlations, see \cite{figueria14, anrev18, JGR21}. Among these, an interesting example is the potential correlation between stellar activity and planetary surface gravity \cite{figueria14}, which links characteristics of the host star and characteristics of the planet (cf. e.g. \cite{figueria14}). Measurement errors are a common occurrence when collecting astronomical data and unfortunately correlation coefficients such as spearman's $\rho$ are quite sensitive even to small perturbations of the observed data. We display an example of the distortion Spearman's correlation suffers under very small additive measurement errors in Figure \ref{Fig:Dist} below.
\begin{figure}[H]
    \centering
    \includegraphics[width=0.66\textwidth]{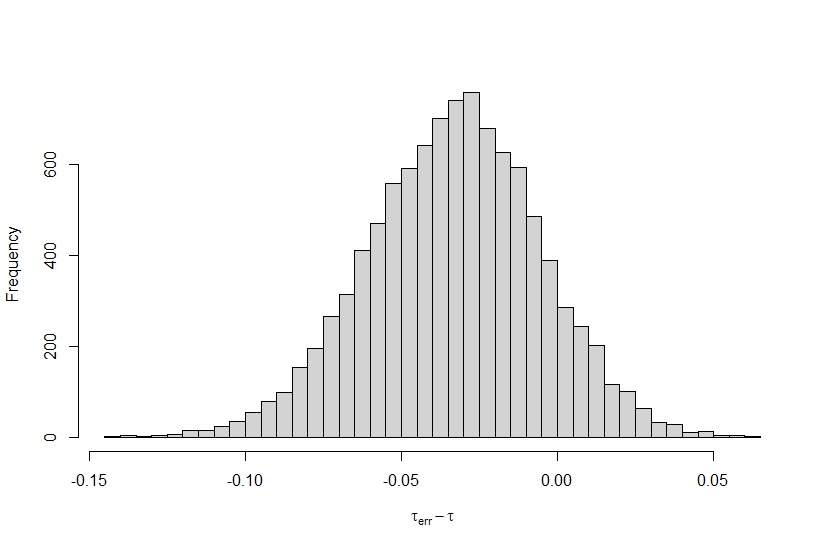}
    \caption{Histogram of differences between Spearman correlations for bivariate data without and with additive error, where we observe $X_i$, resp.  $Z_i =X_i +\varepsilon_i$ \eqref{IntroModel},  with $p=2$, $X$ bivariate normal with correlation$\epsilon_i$ either 0 or a bivariate Laplace distribution with variances equal to $0.05$ and uncorrelated marginals.  We sampled 10000 times at sample size 100, calculated the Spearman correlation without and with error, i.e. for $X_i$ resp. $Z_i$, and recorded their difference.}
    \label{Fig:Dist}
\end{figure}
In the case of the standard (Pearson) correlation the need to account for measurement error is well known \cite{Spearman1904} and a recent review of some available methods can be found in \cite{Saccenti2020}. We do not know of any reference treating general (rank based) correlation coefficients, but in the case of the Pearson and Kendall rank correlation \cite{Kitagawa2018} analyze the first order bias under additive measurement error. They provide a bias correction that requires an estimate of the covariance structure of the latent variables which is often not feasible in practice. They also do not provide an analysis or correction for the impact of measurement errors on the width of standard confidence intervals. Filling these gaps and providing the necessary tools for reliable and flexible inference regarding these correlations is the foremost concern of this work. \\

\textbf{Our Contributions:}\\
In all three references \cite{figueria14, anrev18, JGR21} correlations are estimated and then combined with a procedure to test whether the characteristics are uncorrelated. No uncertainty quantification (i.e. confidence intervals) is provided.  We will contribute to this task in two ways, each accounting for a deficiency of the previous setup. First we note that one often only observes a noisy version $Z = X+\epsilon$ of the desired quantities, here $\epsilon$ is a noise term with a known distribution, the details of which we discuss further below. Ignoring the measurement error results in unreliable inference and so far this has not been accounted for in the available literature. As many correlation coefficients can either be written as or approximated by a $U$-statistic our new methods are able to provide reliable inferential methods even in this setup. In addition we also provide confidence intervals for the estimated parameter that also account for the additive measurement error.
Second, on account of the population correlations rarely being exactly zero, it is well known that in almost any combination of interest of quantities a significant non-zero correlation is detected provided that the sample is sufficiently large and the applied test is consistent. While sometimes even the mere existence of a nonzero correlation may be of scientific interest, it is often the case that only sufficiently large correlations are of practical relevance. This motivates testing relevant hypotheses of the form 
\begin{align}
    \label{hyp}
    H_0(\Delta):\, |\rho|\leq \Delta  \quad vs  \quad H_1(\Delta):\, |\rho|> \Delta,
\end{align} 
where  correlations that are smaller than a given threshold $\Delta$ are discarded.  $\Delta$ can be either specified by the user based on practical considerations or be chosen in a data dependent way, thereby yielding a measure of evidence for/against the existence of a (non-neglible) correlation (see the discussion following Theorem \ref{t3}). The hypotheses \eqref{hyp} therefore offer a more flexible framework that is focused on finding practically significant correlations with a given statistical significance instead of merely detecting any and all nonzero correlations, no matter how small. Similar perspectives have, for instance, been taken in \cite{dette2024} and \cite{bastian2024}  \\

\underline{From a technical perspective} we establish that, given a $U$-statistic of degree 2 with associated kernel $k$ and expected value $\theta=\E[k(Z_i,Z_j)]$, we can construct a deconvolution based estimator $\hat \theta$ that enjoys a central limit theorem of the form
\begin{align}
\label{conv}
    \sqrt{n}h^{1+\beta}(\hat \theta-\E[\hat \theta])\rightarrow G
\end{align}
where $G$ is some Gaussian whose variance depends on $f$ and $k$ and $h$ is a bandwidth parameter. If we allow for undersmoothing bandwidths one may replace  $\E[\hat \theta]$ by $\theta$. Based on this result we construct a test for the hypotheses \eqref{hyp} that relies on a bootstrap procedure to procure the (data dependent) quantiles of $G$. We provide theoretical justifications for \eqref{conv} as well as for the bootstrap procedure.
We note that the derivation of \eqref{conv} is complicated by several issues. Contrary to \cite{bissantz} we can not use (weighted) strong approximations as the available results for the two dimensional case are too restrictive for our purposes. We therefore rely on a poissonization approach similar to \cite{Rosenblatt75}, which in turn is complicated by the fact that the kernels used in kernel deconvolution estimators are usually unbounded in the spatial domain, requiring a more delicate approach to certain bounds.

We also construct confidence bands for the parameter $\theta$ which are often useful in applications. Extensions of our results to higher order $U$-statistics are a straightforward but very cumbersome matter and are therefore omitted as they shed little additional insight into the nature of the problem we consider.  \\

\textbf{Further Related Literature:}\\
 
We first give some general references on deconvolution and its theoretical properties. Several nonparametric estimators for $f$ are available in the deconvolution setting, in particular there are kernel-based estimators (e.g. \cite{sc90}), estimators based on wavelets (\cite{pv99}) and iterative methods (\cite{hm04}). Here we will restrict ourselves to kernel estimators, see Section \ref{results}. For dealing with the deconvolution problem in the context of general statistical inverse problems see \cite{ro99}. It is well known that the minimax rate of convergence of estimators of the true density in such models depends sensitively on the tail behaviour of the characteristic function(s) $\Phi_{\psi}$ of the errors $\epsilon_i$ (cf. \cite{fan91a}). In many cases, results are obtained using the assumption that $\Phi_{\psi}(t)$ never vanishes. If $|\Phi_{\psi}(t)|$ is of polynomial order $|t|^{- \beta}$ for some $\beta >0$ as $|t| \to \infty$ the problem is called ordinary smooth and if  $|\Phi_{\psi^\kappa}(t)|$ is of exponential order $|t|^{\beta_0} e^{-|t|^{\beta}/\gamma}$, $\beta, \gamma>0$, super smooth. Here $| \cdot |$ denotes both the Euclidean norm on $\R^p$ and the absolute value on $\R$. Examples for ordinary smooth deconvolutioon are Laplace and Gamma deconvolution, and for the super smooth case normal and Cauchy deconvolution. A class of examples for ordinary smooth multivariate distributions is given in \cite{gs04}.\\

\underline{Regarding $U$-statistics in the deconvolution setting} we found that there exists barely any literature. To the best of our knowledge the general problem of estimation and inference regarding parameters expressible as $U$-statistics when the data suffers from additive measurement error is as of yet untreated. For the special case of certain rank statistics \cite{Kitagawa2018} find a formula for the bias incurred by additive measurement error that depends on parameters of the distribution of the latent variables. They propose a bias correction that depends on typically inaccessible parameters but do not provide inferential guarantees.\\

For a general development of $U$-statistics without measurement error we refer the interested reader to the seminal paper of  \cite{Hoeffding1948} and to \cite{lee2019} for a comprehensive summary of standard $U$-statistics theory. For robustification of $U$-statistics against heavy tails we refer the interested reader to \cite{Minsker2020}. \\

\textbf{The structure of the paper is as follows.}\\
In Section \ref{results} we present our data model and our main results. Section \ref{simulations} presents a simulation study of the results and Section \ref{hotjupiter} the results for the hot Jupiter data. Proofs are deferred to Appendix \ref{proofsappendix}.

\section{Results}
\label{results}

In this Section we present our estimator and our main results. Let $((X_1,Y_1),...,(X_n,Y_n))$ be a sample of independent identically distributed and paired observations with bivariate density $f$. We observe
\begin{align}
\label{additivemodel}
    (Z_{i1},Z_{i2})=(X_i,Y_i)+(\epsilon_{i,1},\epsilon_{i,2})
\end{align}
where $\epsilon_1,...,\epsilon_n$ is a sequence of iid bivariate noise variables with known density $\psi$. We denote the density of the perturbed sample $(Z_{11},Z_{12})$ by $g$.\\

We are interested in inference regarding a parameter $\theta$ that can be expressed as the expected value of a $U$-statistic of the latent sample $((X_1,Y_1),...,(X_n,Y_n))$. For the sake of notational brevity we will restrain ourselves to $U$-statistics of order 2, extension of the results to higher orders is a straightforward but cumbersome matter. Recall that to each $U$-statistic we associate a symmetric kernel $k:\R^{2\times 2}\rightarrow \R$ such that
\begin{align}
\label{theta}
    \theta=\E[k((X_1,Y_1),(X_2,Y_2))]=\int_{\R^2}\int_{\R^2}k(x,y)f(x)f(y)dxdy~.
\end{align}
As indicated in the introduction (see \eqref{hyp}) we want to test relevant hypotheses of the form
\begin{align}
    \label{UHyp}
    H_0(\Delta):|\theta|\leq \Delta \quad \text{vs} \quad H_1(\Delta):|\theta|>\Delta
\end{align}
to detect practically relevant deviations of $\theta$ from 0. Naturally the choice of $\Delta$ is of major importance, in some cases a natural choice can be identified by the practitioner based on subject knowledge. For cases where no such knowledge is available we provide a data dependent procedure to choose $\Delta$, further discussion of which can be found after the statement of Theorem \ref{t3} and its corollary.\\

To facilitate inference we need a suitable estimator of $\theta$. Equation \eqref{theta} suggests using
 \begin{align}
     \hat \theta=\int_{\R^2}\int_{\R^2}k(x,y)\hat f_n(x)\hat f_n(y)dxdy~.
 \end{align}
where $\hat f_n$ is a suitable estimator of $f$. We postpone discussing computability issues such as how to calculate this 4-dimensional integral to the next section. In the following we shall use a nonparametric kernel estimator for $f$ and to that end we need some additional notation. For any function $d:\R^{p}\rightarrow \R$ we denote its fourier transform by
\begin{align}
    \Phi_d(t)=\int_{\R^{p}} \exp(i\langle t,x \rangle )d(x)dx~,
\end{align}
and additionally let
\begin{align}
    \Phi_n(t)=n^{-1}\sum_{i=1}^ne^{i\langle t, Z_i\rangle}
\end{align}
be the empirical characteristic function of $Z_1,...,Z_n$. Further let $K$ be some kernel, we then denote for some bandwidth $h$ by $\hat f_n(x)$ the deconvolution estimator given by
\begin{align}
    \hat f_n(x)=\frac{1}{4\pi^2}\int_{\R^2} \exp(- i\langle t,x \rangle)\Phi_K(ht)\frac{\hat \Phi_n(t)}{\Phi_\psi(t)}dt\label{eq:est}
\end{align}
which can alternatively be rewritten as
\begin{align}
    \hat f_n(x)=\frac{1}{nh^2}\sum_{k=1}^nK_n((x-(X_k,Y_k))/h)
\end{align}
where the kernel $K_n$ is given by
\begin{align}
K_n(x)=\frac{1}{4\pi^2}\int_{\R^2} \exp(- i\langle t,x \rangle)\frac{ \Phi_K(t)}{\Phi_\psi(t/h)}dt ~.
\end{align}

We make the following assumptions to facilitate our theoretical analysis.
\begin{assumption}
\ \\
\begin{enumerate}
    \item [(A1)] The Fourier transform $\Phi_K$ of K is symmetric, $m\geq 3$ times differentiable and
supported on $[-1, 1]^2, \Phi_K(t) = 1$ for $t \in [-1, 1]^2, c > 0,$ and $|\Phi_K(t)|\leq 1$.   
    \item [(A2)] We have for some $\beta,C>0$ that
    \begin{align}
    \label{Charfraction}
        \frac{\Phi_\Psi(t)}{\norm{t}_2^{-\beta}}\rightarrow C
    \end{align}
    when $t \rightarrow \infty$.
    \item[(A3)] The second derivatives of $f$ are integrable. Additionally we assume that 
    \begin{align}
        \int_{\R^2}\sqrt{F(t)(1-F(t)}dt<\infty
    \end{align}
    \item [(A4)] We have
    \begin{align}
    \label{Assump4}
        \int_{\R^2}|K_{\infty}(z)-h^\beta K_n(z)|dz=o\Big(\frac{1}{\sqrt{n}h^{-3}}\Big)
    \end{align}
    where
    \begin{align}
        K_{\infty}(x):=\frac{1}{C4\pi^2}\int_{\R^2}\exp(-i\langle t,x \rangle)\norm{t}^\beta_2\Phi_K(t)dt
    \end{align}
    and it holds that
    \begin{align}
        \sqrt{n}h^{1+\beta}\rightarrow \infty\\
        \sqrt{n}h^{3+\beta}\rightarrow 0
    \end{align}
    \item [(A5)] The kernel $k$ is a bounded function.
\end{enumerate}
Note that these assumptions imply $h^{\beta}|K_n(z)|\lesssim (1+\norm{z})^{-m}$ for by the Riemann Lebesgue Lemma. The same also holds true for all derivatives of $K_n$. We also note that the kernel $K_n$ is symmetric.
\end{assumption}

Kernels satisfying (A1) are called flat top kernels and possess the favorable property of achieving optimal bias properties irrespective of the smoothness of the density to be estimated. One could also use, for instance, a Gaussian kernel at the expense of more laborious proofs. Instead of assumption (A2) one may instead require that $\Phi_\Psi$ is proportional to a positive semi-definite quadratic form of $(t_1,t_2) \in \R^2$ without changing the presented results except for notational accommodations. The first part of assumption (A3) is a mild regularity condition on the unknown density $f$ while the second part is satisfied whenever $\norm{(X,Y)}$ has finite $(2+\delta)$ moments for some $\delta>0$. Equation \eqref{Assump4} in Assumption (A4) is similar to Assumption 2 in \cite{bissantz} and can therefore be considered as a technical refinement of \eqref{Charfraction}. It holds, for instance, for the Laplace distribution.  Assumption (A5) is always fulfilled for the dependence measures we consider in this paper. We also remark that it can be weakened to a moment assumption at the cost of more laborious proofs and more involved statements of the results.\\

Our first result considers the asymptotic convergence of the estimator $\hat \theta$. Its proof (and the proof of every other Theorem and Corollary in the main text) can be found in the appendix.

\begin{theorem}
\label{t1}
     Under assumptions (A1) to (A5) we have that
    \begin{align}
        \sqrt{n}h^{1+\beta}(\hat \theta -\theta){\rightarrow} \mathcal{N}(0,4\sigma^2) 
    \end{align}
    where 
    \begin{align}
       \sigma^2:= \int_{\R^2}k_y(x)^2f(x)dx\int_{\R^2}K_\infty(x)^2dx~.
    \end{align}
    and $k_y(x)=\int_{\R^2}k(x,y)f(y)dy$.
\end{theorem}

As $k$ and $K$ are known quantities and a consistent estimator of $f$ is available we may construct an inferential procedure and confidence bands based on a suitable estimate $\hat \sigma^2$ of $\sigma^2$. To be precise we consider the test statistic
\begin{align}
    \hat T_{n,\Delta}=\sqrt{n}h^{1+\beta}(|\hat \theta|-\Delta)~,
\end{align}
and the decision rule
\begin{align}
    \text{"Reject $H_0(\Delta)$ if $\hat T_{n,\Delta}>q_{1-\alpha}$"}
\end{align}
where $q_{1-\alpha}$ is the $(1-\alpha)$-quantile of $\mathcal{N}(0,4\hat \sigma^2)$. It is easy to see that this decision rule  yields a consistent asymptotic level $\alpha$ test.
Unfortunately this approach may not perform well in many practical situations due to the limited number of observations available in combination with the rather slow convergence rate $\sqrt{n}h^{1+\beta}$. This motivates the construction of a bootstrap scheme which we pursue in the next subsection. One may similarly use the statistic 
\begin{align}
    \hat T_n=\sqrt{n}h^{1+\beta}\hat \theta
\end{align}
to test the classical hypotheses $\theta=0$ and we note that all results we present are, suitably modified, also true in this setting.

\subsection{A Bootstrap Procedure}
First we give some additional assumptions that we will require for the theoretical analysis of the bootstrap.

\textbf{Bootstrap Assumptions}
\begin{itemize}
    \item[(B1)] $m=\infty$ in (A1) 
    \item[(B2)] There exists a $\kappa>0$ such that $X_i$ and $Y_i$ have finite moments of order $\kappa$. Additionally we require that
    \begin{align}
        \frac{n^{2/\kappa+\delta}}{\sqrt{n}h^{1+\beta}}=o(1)~.
    \end{align}
    holds for some $\delta>0$. 
\end{itemize}
Assumption (B1) serves to facilitate concise proofs and can in principle be dropped, we chose not to do this as this already includes a sufficient selection of suitable kernels $K$ such as flat top kernels.  Assumption (B2) is used to obtain an empirical analogue of assumption (A3) and requires more finite moments the smoother the density $\psi$ is, i.e. the larger $\beta$ is.\\

To construct the bootstrap test statistic let $Z_1^*,...,Z_n^*$ be drawn with replacement from $Z_1,...,Z_n$ and let
\begin{align}
    \hat f_n^*&= \frac{1}{nh^2}\sum_{k=1}^nK_n((x-Z_i^*/h)\\
    \hat \theta^*&=\int_{\R^2}\int_{\R^2}k(x,y)\hat f_n^*(x)\hat f_n^*(y)dxdy
\end{align}
denote the estimates $\hat f_n$ and $\hat \theta$ calculated on the sample $Z_1^*,...,Z_n^*$. In the appendix we prove the following convergence result.
\begin{theorem}
\label{t3}
    Assume (A1) to (A5) as well as (B1) and (B2). Then we have that
    \begin{align}
        \sqrt{n}h^{1+\beta}(\hat \theta^* - \hat \theta) \rightarrow \mathcal{N}(0,4\sigma^2)
    \end{align}
    conditionally on $Z_1,...,Z_n$ in probability.
\end{theorem}

As a consequence of this result we may therefore base our inferential method on the bootstrap quantiles $q_{1-\alpha}^*$ of $\sqrt{n}h^{1+\beta}(\hat \theta^*-\hat \theta)$. We record the properties of the resulting test as a corollary.
\begin{corollary}
Assume (A1) to (A5) as well as (B1) and (B2). Then 
\begin{align}
\label{boottest}
    \lim_{n \rightarrow \infty}\p(\hat  T_{n,\Delta}\geq q^*_{1-\alpha})=\begin{cases}
        0 \quad &\theta<\Delta \\
        \alpha \quad &\theta=\Delta\\
        1 \quad &\theta>\Delta
    \end{cases}
\end{align}
\end{corollary}
Note that the quantile $q_{1-\alpha}^*$ does not depend on the choice of $\Delta$, combined with the fact that the statistic $T_{n,\Delta}$ is a monotone function of $\Delta$ we therefore have that for $\Delta_1>\Delta_2$ the implication
\begin{align}
     \hat T_{n,\Delta_1}>q^*_{1-\alpha} \implies \hat T_{n,\Delta_2}>q^*_{1-\alpha}
\end{align}
holds. By the sequential rejection principle we may therefore test for multiple $\Delta$ without inflating the type 1 error, yielding a data dependent choice of $\Delta$ given by
\begin{align}
    \hat \Delta_{\min}=\min\{\Delta : \hat T_{n,\Delta}\leq q^*_{1-\alpha}\}\lor 0
\end{align}
i.e. the minimal $\Delta$ for which $H_0(\Delta)$ is not rejected. In this way we can sidestep the issue of threshold selection while simultaneously providing a measure of evidence for or against a relevant deviation of $\theta$ from 0.\\

\textbf{Remark: Confidence Intervals}\\
As an alternative and/or additionally to the statistical test, we also propose to determine confidence intervals. I.e. given the estimate $\hat \theta$ and its bootstrap realizations $\hat \theta^*_i$ for $i=1,...,B$ we define the variance estimator
\begin{align}
    \tilde{\sigma}^2 := \frac{1}{B} \sum_{i=1}^B \left(\hat\theta^*_i-\hat\theta\right)^2
\end{align}
As a consequence of Theorems \ref{t1} and \ref{t3} a two-sided asymptotic $(1-\alpha)$-confidence interval is then given by
\begin{align}
\label{ConfInt}
    \left[\hat \theta -z_{1-\frac{\alpha}{2}}\cdot \sqrt{ \tilde{ \sigma}^2},\,\hat \theta +z_{1-\frac{\alpha}{2}}\cdot \sqrt{ \tilde{ \sigma}^2}\right].
\end{align}

where $z_{1-\alpha}$ is the $(1-\alpha)$ quantile of the standard normal distribution. Of course a construction based on $q^*_{1-\alpha}$ is also feasible, but we have found that it performs worse in finite samples on synthetic data and therefore omit it.

\section{Simulations}
\label{simulations}
In this section we present the results of a simulation study of our proposed method. First we describe the general setup of the simulations. Then we describe in more detail the procedure of our method and the results of the study.

\subsection{General setup of the simulations}

In our simulations we consider Kendall's $\tau$ correlation with associated kernel $k$ given by
\begin{align}
    k(x,y)=\1(x_1-y_1)\1(x_2-y_2)~.
\end{align}
Motivated by its widespread practical use we also perform simulations for the Spearman correlation which is, up to an asymptotically negligible term, associated to the kernel
\begin{align}
    k(x,y,z)=\1(x_1-y_1)\1(x_2-z_2)
\end{align}
for which our results do not directly apply, but can be extended with some tedium. \\

We consider two different scenarios for the additive error. In more detail we consider
\begin{align}    
    Z_i^j & \sim \mathcal{N}(0,\mathbf\Sigma_j), \quad i=1,...,n\\ 
    \epsilon_i^j &\sim Lap(0,\mathbf T_j), \quad i=1,...,n
\end{align}
where the covariance matrices are given by
\begin{align}
\mbox{Model 1: } &  
    \label{Model1}
  \mathbf{\Sigma_1} = \left(\begin{array}{cc}
                  1 & \sigma_1 \\
                  \sigma_1 & 1\end{array}\right), 
  \quad \quad \mathbf{T_1} = \left(\begin{array}{cc}
                  0.05 & 0 \\
                  0 & 0.05\end{array}\right)\\  
  \mbox{Model 2: } &
  \label{Model2}
  \mathbf{\Sigma_2} = \left(\begin{array}{cc}
                  1& \sigma_2 \\
                  \sigma_2 & 3\end{array}\right), 
  \quad \quad\mathbf{T_2} = \left(\begin{array}{cc}
                  0.061& 0  \\
                  0 & 0.0025\end{array}\right)
\end{align}
These models cover the settings where the true density is either perfectly radially symmetric or very asymmetric. The covariances $\sigma_1$ and $\sigma_2$ are parameters which we will vary to simulate both data under the null and the alternative. The choices $\mathbf{Sigma}_2$ and $\mathbf T_2$ yield a model with less regular shape, which is motivated by the empirical results for the hot Jupiter data that we will discuss further below. \\

Before we proceed with presenting the results some discussion of the numerical implementation is warranted. 
\begin{enumerate}
    \item To calculate the estimators we use the fast fourier transform (FFT) and inverse FFT implemented in scipy 1.10.1 and  1.11.4 \cite{scipy} with a grid size of 512x512 for model 1 and 1024x1024 for model 2, where model 2 requires a larger grid size due to its more complicated shape and the large difference of the variances in the covariance matrix $T_2$.
    \item Naively calculating the 4 dimensional integral in the definitions of $\hat \theta$ and $\hat \theta^*$ has $O(n^4)$ runtime. Due to the lack of knownledge about the structure of the integrand in the definition of $\hat \theta^*$ it is also difficult to apply more specific algorithms with better runtimes. We therefore use Monte Carlo integration based on sampling from the distributions defined by $\hat f_n(x)$ and $\hat f_n^*(x)$ 2500 times to avoid this computational issue. To avoid further numerical problems (i.e. possible negative values or positive values at grid points very far from the center of the grid) we threshold the densities by setting the density to 0 at all grid points where their value lie below $1/1000$ of their maximum value and then normalize to guarantee an integral of 1 afterwards.
    \item As the sampling from $\hat f_n$ occurs on a grid the calculation of rank correlations suffers from the presence of ties if the grid is not chosen suitably large. For example choosing smaller but practically feasible grids leads to a substantial bias in the estimates of Kendall's $\tau$. To avoid this issue we perturbed the samples from $\hat f_n$ and $\hat f_n^*$ by samples from a uniform distribution on the interval $[0,10^{-6}]$. The results from \cite{Kitagawa2018} indicate that the bias that is introduced by this correction is of order $10^{-12}$ and therefore negligible.    
    \item For data where the error variances, say $s^2_1$ and $s_2^2$, are heterogeneous it is advisable to use different bandwidths along each coordinate. Consequently we consider
    \begin{align}
        \Phi_K(h(s_2s_1^{-1}t_1,t_2))
    \end{align}
    instead of $\Phi_K(ht)$ in \eqref{eq:est}. This has no impact on the asymptotic theory beyond notational inconvenience.    
\end{enumerate}

\textbf{Choice of bandwidth:}\\
Clearly the selection of the bandwidth $h$ is crucial for the performance of the proposed method. In our simulations we have used a method similar in spirit to that used by \cite{bissantz} which is based on the following observation: Denoting by $h_{opt}$ the bandwidth that minimizes the $L^2$ distance between $\hat f_n$ and $f$ we observe that for over-smoothing bandwidths $h>h_{opt}$ the behavior of the estimator changes only slowly with increasing bandwidth as more and more features of the distribution get smoothed out. Conversely, for bandwidths $h<h_{opt}$, we include frequencies of greater magnitude in the estimator \eqref{eq:est}. For these frequencies the characteristic function $\Phi_\psi$ takes on small values, exacerbating the fluctuations in $\hat \Phi_n$ and thereby leading to increasingly strong fluctuations in the estimator that are not present in the true density, yielding a steep increase in the $L^2$ error when the bandwidth decreases below $h_{opt}$. Consequently, letting $\hat f_n(h)$ denote the estimator $\hat f_n$ for the bandwidth $h$, we expect that the quantity 
    \begin{align}
        \norm{\hat f_n(h_1)-\hat f_n(h_2)}
    \end{align}
    behaves as follows: it will be relatively small for $h_1,h_2>h_{opt}$, mostly depending on the distance $|h_1-h_2|$ between the bandwidths. As soon as at least one of the bandwidths falls below $h_{opt}$ it has a sharp spike as the estimators starts overfitting. We therefore choose the bandwidth as follows: Consider a log-spaced sequence $h_1,...,h_{m}$ of bandwidths and define
    \begin{align}
        \hat h_{opt}=\arg \min_{1 \leq h_i \leq m-1}\underbrace{\norm{\hat f_n(h_i)-\hat f_n(h_{i+1})}_2}_{=:D_i(h_i)}
    \end{align}
    Figure \ref{fignorm500} illustrates empirically that this choice coincides fairly well with the minimum of the global mean square error $\norm{f-\hat f_n(h)}$ for the case of model 2 with $n=500$ as an example.

\begin{figure}[H]
\begin{center}   
\includegraphics[width=4.5in]{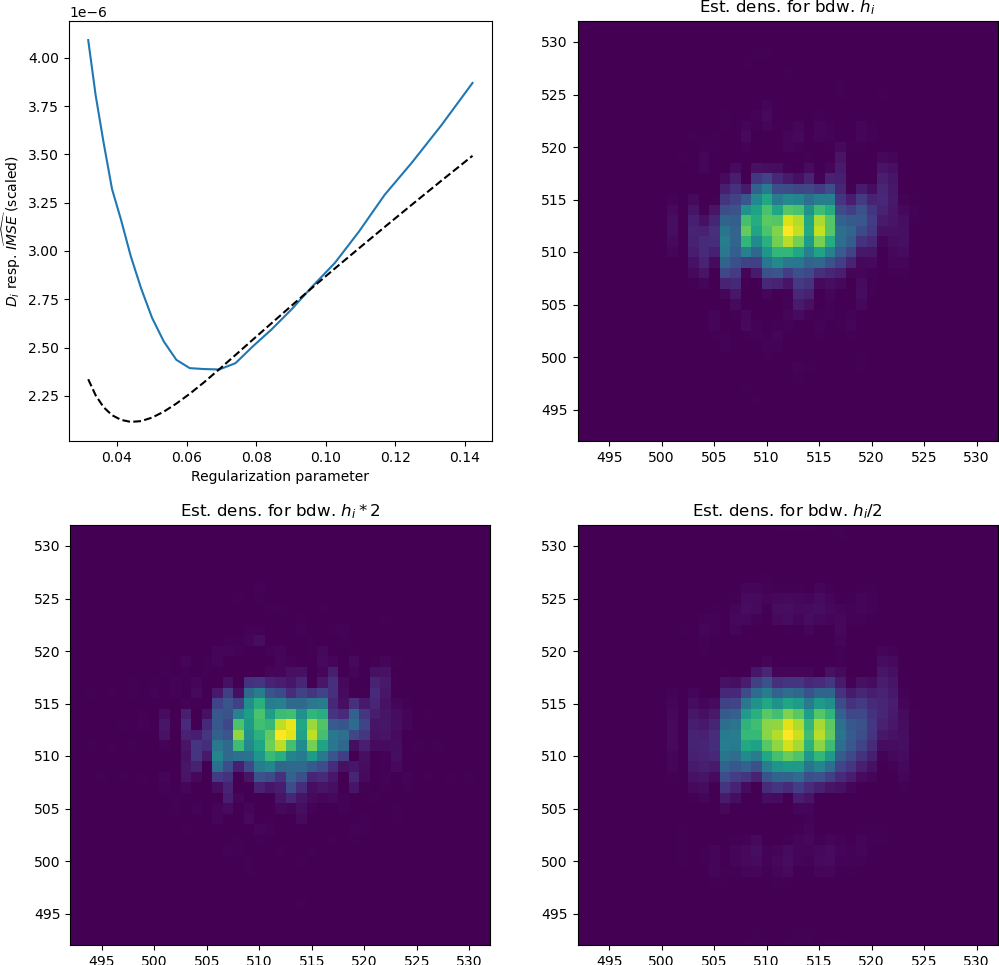}
\end{center}
\caption{\label{fignorm500} Results for Model\eqref{Model2} and $n=500$. We plot $D_i(h_i)$ and $\widehat{IMSE}(h_i)$ against $h_i$. The other graphs contain heatmaps of the estimated densities for the regularization parameters $\frac{j}{2}h_i^{\rm opt}, j=1,2,3$ where $h_i^{\rm opt}$ minimizes $(D_i)_{i=1,...,m}$.}
\end{figure}

\subsection{Simulation results}

We consider the sample sizes $n=100$ and $500$ and apply the procedures \eqref{boottest} and \eqref{ConfInt} to $250$ datasets to determine empirical rejection and coverage rates for each Model and choice of $\sigma_i$. For each dataset we generated $250$ bootstrap samples to determine the critical value $q_{1-\alpha}^*$. Regarding the choice of $\Delta$ we consider 
\begin{align}
    \Delta_1&=0.333\\
    \Delta_2&=0.037
\end{align} for Kendall's $\tau$ in Model 1 and 2, respectively.
For Spearman's $\rho$ we consider
\begin{align}
    \Delta_1&=0.483\\
    \Delta_2&=0.056~,
\end{align}
in Model 1 and 2, respectively. Our choices for $\sigma_1$ and $\sigma_2$ together with the respective values of Kendall's $\tau$ and Spearman's $\rho$ can be found in the table below. 


\begin{table}[H]
  \centering
   \setlength{\tabcolsep}{4pt}
  \begin{tabular}{@{} c|c|c|c|c|c|c @{}} 
  Model & \multicolumn{3}{c|}{Model 1} & \multicolumn{3}{c}{ Model 2}\\
  \hline
    &\small $H_0$ &\small   $H_1^A$ & \small $H_1^B$ &\small $H_0$ &\small   $H_1^A$ & \small $H_1^B$\\ 
  \hline
    \small$\sigma$&\small 0.5 &\small   0.55 & \small 0.7 &\small 0.1 & \small  0.11 & \small 0.23 \\ 
    \hline
    \small Kendall's $\tau$  & 0.333 & 0.370 & 0.493 & 0.037 & 0.040 & 0.085 \\ 
    \small Spearman's $\rho$ & 0.483 & 0.531 & 0.682 & 0.056 & 0.061 & 0.127 \\ 
  \end{tabular}
  \caption{Rank correlations for the Models \eqref{Model1} and \eqref{Model2} for different choices of $\sigma_1$ and $\sigma_2$}
  \label{tab:corralternatives}
\end{table}

For each model the smallest $\sigma$ corresponds to the boundary of the null hypotheses $H_0(\Delta)$ while the second and third choices are part of the alternative $H_1(\Delta)$. 
We have selected the initial bandwidths for the simulations with the method discussed in the previous section. For a discussion of the robustness properties of this choice see the discussion of table \ref{tab:labeldatad} further below). We present the results in tables \ref{tab:labelkendall} and \ref{tab:labelspearman} below.

\begin{table}[H]
  \centering
  \begin{tabular}{c|c|c|c|c|c|c|c}
    n & h ($H_0$) & $H_0$ & h ($H_1^A$) & $H_1^A$ & h ($H_1^B$) & $H_1^B$& CI\\ 
    \midrule
    100 &  0.064 & 0.032 & 0.0507 & 0.024 & 0.0629 & 0.384 & 0.896\\ 
    500 & 0.0845 & 0.044 & 0.0737 & 0.18 & 0.0915 & 0.98 & 0.948 \\ 
    100 & 0.045 & 0.068 & 0.0399 & 0.132 & 0.0448 & 0.468 & 0.952 \\ 
    500 & 0.072 & 0.096 & 0.0693 & 0.272 & 0.0693 & 0.88 & 0.98\\ 
  \end{tabular}
  \caption{\label{tab:labelkendall}Empirical rejection rates of the test $\eqref{boottest}$ for Kendall's $\tau$ for Model \eqref{Model1} (rows 1-2, $\Delta=0.333$) and Model \eqref{Model2} (rows 3-4, $\Delta=0.037$) . Column CI contains empirical coverage probabilities for two-sided $95\%$ confidence intervals. Bandwidths shown are for the simulations under $H_0$ (on which CI coverage rates are also based) and $H_1^A$, $H_1^B$, respectively.}
\end{table}

\begin{table}[H]
  \centering
  \begin{tabular}{c|c|c|c|c|c|c|c}
    n & h ($H_0$) & $H_0$ & h ($H_1^A$) & $H_1^A$ & h ($H_1^B$) & $H_1^B$& CI\\ 
    \midrule
    100 & 0.064  & 0.02  & 0.0507 & 0.02 & 0.0629 & 0.256 & 0.896 \\
    500 & 0.0845 & 0.048 & 0.0737 & 0.176 & 0.0915 & 0.976 &  0.948 \\
    100 & 0.045 & 0.064 & 0.0399 & 0.128 & 0.0448 & 0.472 &  0.948 \\
    500 & 0.072 & 0.084 &  0.0693 & 0.28 & 0.0693 & 0.892 & 0.972 \\
  \end{tabular}
  \caption{\label{tab:labelspearman} Empirical rejection rates of the test $\eqref{boottest}$  for Spearman's $\rho$ for Model \eqref{Model1} (rows 1-2, $\Delta=0.483$) and Model \eqref{Model2} (rows 3-4, $\Delta=0.056$). Column CI contains empirical coverage probabilities for two-sided $95\%$ confidence intervals. Bandwidths shown are for the simulations under $H_0$ (on which CI coverage rates are also based) and $H_1^A$, $H_1^B$, respectively.}
\end{table}

The method thus performs very well for the symmetric model and seems to have a slightly inflated size for the more complex radially asymmetric model as can been in particular from the $H_0$ column in tables \ref{tab:labelkendall} and \ref{tab:labelspearman} and taking into account that the width of an approximate asymptotic $95\%$ confidence interval for an estimate of the probability of a Bernoulli variable with true probability of $5\%$ is
\[\pm z_{0.975}\cdot\sqrt{\frac{0.05\cdot (1-0.05)}{250}}\approx 2.7\%.\]

The confidence intervals tend to be slightly conservative. Confidence intervals based on $\overline{\hat\tau^{B,\cdot}}$ instead of $\hat \tau$ are a natural alternative that we omitted on account of them being overly liberal. \\

Finally we consider the sensitivity of the procedure with regards to the bandwidth choice in table \ref{tab:labeldatad}. To that end we consider the bandwidths $(1\pm0.1)\hat h_{opt}$ and investigate how the approximation of the nominal level in the settings of the tables \ref{tab:labelkendall} and \ref{tab:labelspearman} is impacted. The choice of $\pm 0.1$ is motivated by the width of the local minimum of $D_i$ in Figure \ref{fignorm500}. We also considered a small number of further bandwidths within a range of $\pm 0.01$ of $h_{opt}$ and observed little variation in the results. 

\begin{table}[H]
  \centering
  \begin{tblr}{@{} c|c|c|c @{}}
    Model & h & level Kendall's $\tau$ & level Spearman's corr.\\ 
    \midrule
    1 & 0.07605  & 0.056  & 0.072 \\
    1 & 0.09295 & 0.048 & 0.048 \\
    2 & 0.0648 & 0.064 &  0.068\\
    2 & 0.0792 & 0.088 & 0.072\\
  \end{tblr}
  \caption{\label{tab:labeldatad}Results for Model \eqref{Model1} (rows 1-2) and Model \eqref{Model2} (rows 3-4) for bandwidths $10\%$ above respectively below the values used in Tables \ref{tab:labelkendall} and \ref{tab:labelspearman}, we consider only the sample size $n=500$.}   
\end{table}
We observe that the approximation of the nominal level seems to be robust to reasonably sized perturbations of the bandwidth, indicating that we may combine the proposed bandwidth selection algorithm and the test to a fully data driven method.


\section{Correlation between stellar activity and planetary surface gravity for Hot Jupiters}
\label{hotjupiter}

Before we present the data set we consider in detail we will make some general remarks about it and the problems we consider. Hot Jupiters are a class of stellar objects that have been observed for the first time fairly recently (see \cite{mayquel95}). While the possibility of their existence had already been considered by \cite{struve52} they were not predicted by any of the common planetary system formation models (see \cite{anrev18}, \cite{JGR21} for more details). To close these gaps gaining an understanding of the mechanism by which they come to be is a natural course of action. It is therefore paramount to identify the physical quantities which are potentially relevant to this mechanism.  A first major step in this direction often is an analysis of the correlations between different physical quantities that are plausible candidates for explanatory mechanism. This avenue of approach is of course not unique to this situation and is pursued in many other fields of science \cite{Russo2007, Cartwright2007} where candidates for causal connections are identified by an analysis of correlations. A high or low correlation may then be a good indicator to determine whether or not further study might be worthwhile.\\ 

As already mentioned in the introduction Hot Jupiters have fairly large masses combined with a small orbit around their parent star. These characteristic make them ideal targets for radial-velocity based detection methods as the induced variability of the radial velocity of the parent star is substantially larger than for more remote planets. Additionally the inclination of their orbital plane is often close to 90 degrees which makes them favorable candidates for detection by the transit method due to the resulting brightness variability of the parent star. Since the discovery of the first Hot Jupiter 51 Peg b in 1995, which also was the first discovery of an exoplanet in general, some few hundred hot Jupiters have been found. \\

In view of the preceding discussion it comes at no surprise that several physical characteristics of Hot Jupiters and their parent stars haven been investigated for correlations. Among these, a classic example is the potential correlation between stellar activity and planetary surface gravity as in \cite{figueria14}. We use the data presented in their study, which consists of data for 108 hot Jupiters with both quantities available. The data is sourced from two works, the first contains data for 39 hot Jupiters \cite{hartman2010} while the second \cite{schneider2011} contains another 69 additional pairs of observations. \\

To apply our methods we need to \textbf{specify the error distribution} in model \eqref{additivemodel}. In the following we will consider 
\begin{align}
    \epsilon_i \sim Lap(0,\Sigma),
\end{align}
where 
\begin{align}
    \Sigma=\begin{pmatrix}
        0.0036 & 0\\ 0 & 0.0025
    \end{pmatrix}
\end{align} 
These values are motivated by the discussion in \cite{figueria14}, who summarizes the available information on the data error from different sources. They claim a difference of order $\pm 0.05$ for the logarithmic planetary surface gravity based on the possible impact of stellar spots on estimation of the planetary circumference. For the stellar activity $\log(R'_{HK})$ the error is apparently difficult to estimate and also missing for some data sources, \cite{figueria14} quote approximate errors for measurements of $\log(R’_{HK})$ to be $0.02-0.1$ dex where available. We thus decided to use $0.06$ dex as a compromise value. We note that using more extreme values of $0.04$ resp. $0.08$ does not have a major impact on the results of our analysis ($90\%$-CI for Kendalls $\tau$ of $(-0.010,  0.185)$ and $(0.011, 0.259)$, resp. for $0.04$ and $0.08$ dex, and $(-0.014, 0.270)$ and $(0.011, 0.372)$, resp. for the Spearman correlation, see the discussion further below for context).

The Laplace distribution is a common choice when modeling multivariate data with tails heavier than a normal distribution that still has finite moments (\cite{kotz01}). It is also a more robust choice in case of misspecification than the normal or other supersmooth distributions \cite{Meister2004}. We also remark that previous studies (cf. e.g. \cite{Hesse1999}) have observed empirically that deconvolution methodology is usually rather robust to distributional assumptions. \\

Due to the relatively low sample size we did not calculate the \textbf{regularization bandwidth $h$} based on the true sample. We instead opted to use  the optimal regularization parameter for $n=100$ in model 2, which assumes the same error distribution but assumes that the latent data are normally distributed. For real data the data distribution can be very irregular, either due to insufficient sample size or true physical characteristics of the density to be estimated. This makes it difficult to be used for simulations, e.g. for our proposed regularization parameter selection method. Here we propose to use mock densities in this case, e.g. a bivariate normal distribution with similar covariance matrix as estimated from the data. 

 Our results for the Hot Jupyter data are shown in Table \ref{tab:hotJupiter}. We have also performed an analysis of the sensitivity of the result on the assumed regularization parameter by repeating the analysis with $1.5\times$ and $2/3 \times$ the chosen bandwidth and observed that the general conclusion of the analysis were not impacted by this.

\begin{table}[H]
  \centering
  \begin{tblr}{@{} c|c|c @{}}
    Kendall's tau (and conf. int.), $\hat \Delta_{\min}$ & $\rho_{sp}$ (and conf. int.), $\hat \Delta_{\min}$ & Reg. parameter\\ 
    \midrule
    $0.115$ $([0.013,\,0.217])$, $0.043$ & $0.164$ $([0.014,\,0.314])$, $0.057$ & $h_i^{\rm opt}$\\
    $0.112$ $([0.011,\,0.213])$, $0.044$ & $0.160$ $([0.013,\,0.308])$, $0.065$ & $h_i^{\rm opt}\times (2/3)$\\
    $0.130$ $([0.013,\,0.248])$, $0.038$ & $0.186$ $([0.013,\,0.359])$, $0.049$ & $h_i^{\rm opt}\times 1.5$\\    
  \end{tblr}
  \caption{Point estimates and $95\%$-confidence bands  of the rank correlations of the Hot Jupiter data set. We also record the largest $\Delta$, for which $H_0(\Delta)$ (see \eqref{hyp}) is rejected,. Results are displayed for the bandwidths $2/3 h_i^{opt},h_i^{opt},3/2h_i^{opt} $.}
  \label{tab:hotJupiter}
\end{table}

Let us compare these results with those obtained in \cite{figueria14}. The authors did not take potential measurement errors into account and find Spearman correlations of $0.45$ and $0.21$ for the first 39 observations and the full data set, respectively. They reject the null of no correlation based on a permutation procedure. While we also find that our confidence bands do not contain 0 (albeit barely so) a closer look at $\hat \Delta_{\min}\sim 0.055$ tells a more complete story - we in fact only have sufficient evidence for a very small spearman correlation. We also note that the correlation between $\log(R'_{HK})$ and $\log(g_p)$ appears to be somewhat difficult to understand from a physical perspective and its significance is not clear \cite{anrev18}. This is consistent with our results which indicate that, while statistically significant in the classical sense, no relevant correlation exists and illustrates 
\begin{enumerate}
    \item[(1)] that in the construction of confidence bounds or inference accounting for observational errors is important even if the observation errors are fairly small on a first glance.
    \item[(2)] that adopting a relevant hypothesis framework facilitates meaningful discussion of the physical implications of statistical conclusions beyond merely deciding whether or not a particular observation is consistent with an independence assumption.
\end{enumerate}

\textbf{Acknowledgements}
This research was partially funded in the course of TRR 391 Spatio-temporal Statistics for the Transition of Energy and Transport (520388526) by the Deutsche Forschungsgemeinschaft (DFG, German Research Foundation).


\ \\

\ \\
\section{Proofs}
\label{proofsappendix}
\textbf{Proof of Theorem \ref{t1}}
\begin{proof}
    From Lemma \ref{pl1} we obtain that
    \begin{align}
        \sqrt{n}h^{1+\beta}(\hat \theta -\E[\hat \theta])=\sqrt{n}h^{1+\beta}(\hat \theta -\theta)+o_\p(1)~
    \end{align}
    Lemma \ref{pl3} yields
    \begin{align}
        \sqrt{n}h^{1+\beta}(\hat \theta -\E[\hat \theta])=2\sqrt{n}h^{1+\beta}\int_{\R^2} k_y(x)(f_n(x)-\E[f_n(x)])dx+o_\p(1)
    \end{align}
    Lemma \ref{pl4} and \ref{pl5} then yield
    \begin{align}
     2\sqrt{n}h^{1+\beta}\int_{\R^2} k_y(x)(f_n(x)-\E[f_n(x)])dx=2\sqrt{n}hS_n+o_\p(1)~.
    \end{align}
    The proof is finished by combining the previous equations with Theorem \ref{t2}. 
\end{proof}

\begin{Lemma}
\label{pl1}
    Assume that (A1) to (A5) hold. Then
    \begin{align}
        |\E[\hat \theta_n]-\theta| \lesssim O(n^{-1}h_n^{-2-2\beta}+h^2)
    \end{align}
\end{Lemma}
\begin{proof}
    We first note that
    \begin{align}
       f_n(x)f_n(y)&=\frac{1}{nh^2}\sum_{k=1}^nK_n((x-(X_k,Y_k))/h)\frac{1}{nh^2}\sum_{k=1}^nK_n((y-(X_k,Y_k))/h)\\
       &=\frac{1}{n^2h^4}\sum_{k\neq j}^nK_n((x-(X_k,Y_k))/h)K_n((y-(X_l,Y_l))/h)\\
       &+\frac{1}{n^2h^4}\sum_{k=1}^nK_n((x-(X_k,Y_k))/h)K_n((y-(X_k,Y_k))/h)       
    \end{align}
    and observe that each summand is bounded by a multiple of $h^{-\beta}$. Using that $k$ is bounded then yields that 
    \begin{align}
         \hat \theta_n&=\int_{\R^2}\int_{\R^2}k(x,y)\frac{1}{n^2h^4}\sum_{k\neq j}^nK_n((x-(X_k,Y_k))/h)K_n((y-(X_l,Y_l))/h)dxdy\\
         &\quad +O((\sqrt{n}h^{1+\beta})^{-2})
    \end{align}
    Take expectations, use Fubini and then independence. After this a similar calculation as above can be used to reintroduce the diagonal sum, thereby yielding that
    \begin{align}
        \E[\hat \theta_n]= \int_{\R^2}\int_{\R^2}k(x,y)\E[f_n(x)]E[f_n(y)]dxdy+O((\sqrt{n}h^{1+\beta})^{-2})~.
    \end{align}
    Consequently we obtain the following bound
    \begin{align}
        |\E[\hat \theta_n]-\theta|&= \Big|\int_{\R^2}\int_{\R^2}k(x,y)\Big(\E[f_n(x)]E[f_n(y)]-f(x)f(y)\Big)dxdy\Big|+O((\sqrt{n}h^{1+\beta})^{-2})\\
        &\leq  \norm{k}_\infty \norm{\E[f_n(x)]E[f_n(y)]-f(x)f(y)}_{1,\R^2\times \R^2}+O((\sqrt{n}h^{1+\beta})^{-2})
    \end{align}
    Thus the we are left with finding bounds for
    \begin{align}
        \norm{\E[f_n(x)]E[f_n(y)]-f(x)f(y)}_{1,\R^2\times \R^2}&\lesssim \norm{\E[f_n(x)]-f(x)}_1
    \end{align}
    That this term is of order $h^2$ follows by a standard analysis as in Theorem 24.1 from \cite{van1996}.
\end{proof}

\begin{Lemma}
\label{pl2}
    Suppose that
    \begin{align}
        \int_{\R^2}\sqrt{F(t)(1-F(t)}dt<\infty
    \end{align}
    Then
    \begin{align}
         \norm{\E[f_n(x)]-f_n(x)}_1 \lesssim O_\p\Big(n^{-1/2}h^{-\beta-1}\Big)
    \end{align}
\end{Lemma}
\begin{proof}
    Note that $\E[f_n(x)]-f_n(x)=h^{-2}(K_n \star (F_n-F))(x)$, we then obtain by partial integration, Young's convolution inequality and Markovs inequality the following bounds
    \begin{align}
        \int_{\R^2}\int_{\R^2}\Big|K_n\Big(\frac{x-y}{h}\Big)d(F_n-F)(y)\Big|dx&=\sum_{i=1}^2\int_{\R^2}\Big|\int_{\R^2}h^{-1}K_{ni}\Big(\frac{x-y}{h}\Big)(F_n(y)-F(y))dy\Big|dx\\
        &\leq\sum_{i=1}^2h\norm{K_{ni}}_1\norm{F_n-F}_1    \\
        &\leq \sum_{i=1}^2h\norm{K_{ni}}_1 O_\p\Big( n^{-1/2}\Big)
    \end{align}
    where $K_{ni}$ denotes the $i-$th partial derivative of $K_n$. $\Phi_K$ having compact support, sufficient smoothness and the fact that $\Phi_\Psi(t) \simeq \norm{t}^{-\beta}$ then yields (due to the Rieman-Lebesgue Lemma) that $\norm{K_{ni}}_1\lesssim h^{-\beta}$ and finishes the proof.    
\end{proof}

\begin{Lemma}
\label{pl3}
     Under assumption (A1)-(A5) we have
    \begin{align}
        sup_{y \in \R^2}|k_{x,n}(y)-k_x(y)|&\lesssim h^{-\beta-1}n^{-1/2}\\
        \Big|\int_{\R^2}(k_{x,n}(y)-k_x(y))(f_n(y)-\E[f_n(y)])dy\Big| &\lesssim h^{-2\beta-2}n^{-1} 
    \end{align}
\end{Lemma}
\begin{proof}
    The first statement follows immediately by noting that 
    \begin{align}
        |k_{x,n}(y)-k_x(y)| \lesssim  \norm{\E[f_n(x)]-f_n(x)}_1     
    \end{align}
    and Lemma \ref{pl2} while the second statement follows from the first and an application of Hölder's inequality to obtain
    \begin{align}
        \Big|\int_{\R^2}(k_{x,n}(y)-k_x(y))(f_n(y)-f(y)dy\Big| \lesssim  \norm{ k_{x,n}(y)-k_x(y)}_\infty \norm{\E[f_n(x)]-f_n(x)}_1    
    \end{align}
\end{proof}

We define for some sequence $a_n \rightarrow \infty$ with $ha_n \rightarrow 0$ the asymptotic Kernel $K_\infty$ and its truncated version $\tilde K_\infty$ by
\begin{align}
    K_\infty&=\frac{1}{C4\pi^2}\int_{\R^2}\exp(-i\langle t,x \rangle  )\norm{x}^\beta \Phi_K(x)dx\\
    \tilde K_\infty&=\frac{1}{C4\pi^2}\int_{\norm{x}\leq a_n}\exp(-i\langle t,x \rangle  )\norm{x}^\beta \Phi_K(x)dx~.
\end{align}
Note that these Kernels are symmetric because $\Phi_K$ is symmetric, in particular their first moments are zero.\\

We further denote their partial derivatives by $\bar K_{\infty,i}$ and $K_{\infty,i}$, respectively. We also define the associated density estimates $f_n^{K_\infty}$ and $\tilde f_n$ by
\begin{align}
    f_n^{K_\infty}=h^{-2}(K_\infty \star F_n)\\
    \tilde f_n=h^{-2}(\tilde K_\infty \star F_n)
\end{align}
\begin{Lemma}
\label{pl4}
    Under assumptions (A1) to (A5) we have
    \begin{align}
     n^{1/2}h\norm{\tilde f_n-h^\beta f_n}_1=o(1)
    \end{align}
\end{Lemma}
\begin{proof}
    We first replace $K_n$ by its asymptotic Kernel $K_\infty$. We note that
    \begin{align}
        \norm{f_n^K-h^\beta f_n}_1 \lesssim \int_{\R^2}\int_{\R^2}\Big|(K_\infty-h^\beta K_n)\Big(\frac{x-y}{h}\Big)\Big|dxdF_n(y)
    \end{align}
    The substitution $z=(x-y)/h$ yields that this is bounded by
    \begin{align}
        h^2\int_{\R^2}|K_\infty(z)-h^\beta K_n(z)|dz:=c_n
    \end{align}
    which by Assumption (A4) fulfills $c_n=o(n^{-1/2}h^{-1})$.\\

    We are left with bounding $\norm{\tilde f_n - f_n^{K_\infty}}_1$. Using similar arguments as in the first step we find the upper bound
    \begin{align}
        h \int_{\norm{z}\geq a_n}|K_\infty(z)|dz
    \end{align}
    Using that $\Phi_K$ has compact support and derivatives of order $m$ we obtain by the Riemann-Lebesgue Lemma that $|K_\infty(z)| \lesssim (1+\norm{z})^{-m}$. Consequently, using polar coordinates, we obtain
    \begin{align}
        \norm{\tilde f_n - f_n^K}_1\lesssim ha_n^{-m+2}
    \end{align}
    which finishes the proof.
\end{proof}

We now apply a poissonization argument to the density estimator based on the truncated asymptotic kernel $\tilde K_\infty$. To that end let $N$ denote a poisson random variable independent of the sequence $\{Z_i\}$ with mean $n$ and define
\begin{align}
    f_n^{po}=n^{-1}h^{-2}\sum_{i=1}^N\tilde K_\infty\Big(\frac{x-Z}{h}\Big)
\end{align}
We note that the resulting empirical measure $F_n^{po}$ has the property that $nF_n^{po}$ is a poisson process on a plane, in particular we have the following properties (compare \cite{Rosenblatt75}).
\begin{align}
\label{p4}
n^k \E \left( d(F_n^{po} - F) \right)^{2k} &= \frac{(2k)!}{k! \, 2^k} (dF)^k + \sum_{j=1}^{k-1} a^{(2k)}_{j,n} \, (dF)^j,\\
n^{(2k+1)/2} \E \left( d(F_n^{po} - F) \right)^{2k+1} &= \sum_{j=1}^{k} a^{(2k+1)}_{j,n} \, (dF)^j,
\end{align}

where

\begin{equation}
a^{(s)}_{j,n} = O \left( n^{-((s/2) - j)} \right)
\end{equation}

for each fixed \( j \), \( s \) with \( j = 1, \ldots, (s/2) - 1 \) and \( (u) \) is the smallest integer greater than or equal to \( u \).

\begin{Lemma}
\label{pl5}
    Under assumptions (A1)-(A5) we have
    \begin{align}
        \E\Big[\norm{\tilde f_n-f_n^{po}}_1\Big]\lesssim n^{-1/2}
    \end{align}
\end{Lemma}
\begin{proof}
    Note that due to the independence of $N$ from the data we have
    \begin{align}
        \E[|\tilde f_n-f_n^{po}|]\leq n^{-1}h^{-2}\E[|N-n|]\E\Big[\Big|\tilde K_\infty\Big(\frac{x-Z}{h}\Big)\Big|\Big]
    \end{align}
    and then use that $\E[|N-n|]\leq \sqrt{n}$ and standard arguments as in the proof of Theorem 24.1 from \cite{van1996} to obtain 
    \begin{align}
        \int_{\R^2}\E\Big[\Big|\tilde K_\infty\Big(\frac{x-Z}{h}\Big)\Big|\Big]dx\lesssim h^2+o(h^2) 
    \end{align}
    due to $\tilde K_\infty$ having (uniformly in $n$) finite second moments. This finishes the proof.
\end{proof}

In preparation for deriving the asymptotic normality of the poissonized statistic
\begin{align}
    S_n=\int_{\R^2}k_y(x)(f^{po}_n(x)-\E[f^{po}_n(x)])dx
\end{align}
we define for a sequence $c_n \rightarrow 0$ with $ha_n=o(c_n)$ the quantities
\begin{align}
    \Delta_j&=4((j+1)ha_n+jc_n)\\
    \Delta_j'&=4(j+1)(ha_n+jc_n)\\
    I_{jk}&=[\Delta_j,\Delta_j']\times[\Delta_k, \Delta_k']
\end{align}
and observe the following
\begin{Lemma}
\label{pl6}
We have that the random variables
\begin{align}
     V_{jk}:= \int_{I_{jk}}k_y(x)\int_{\R^2}\tilde K_\infty\Big(\frac{x-y}{h}\Big)d(F^{po}_n-F)(y)dx, \quad j,k \in \mathbb{Z}
\end{align}
  are independent.
\end{Lemma}
\begin{proof}
Using the fact that $\tilde K_\infty\Big(\frac{x-y}{h}\Big)$ is supported on $[x-a_nh,x+a_nh]$ ("+" and "-" are to be understood coordinate wise) we have that $V_{jk}$ can be expressed as an integral of a measurable function over $J_{jk}=(I_{jk}+a_nh) \cup (I_{jk}-a_nh)$ with respect to $(F_n^{po}-F)$ (here "+" and "-" denote Minkowski sums/differences). By construction of $I_{jk}$ neighboring sets are separated by strips of width $4ha_n$ so that $J_{jk}$ are all pairwise disjoint. Combining this with the fact that $nF_n^{po}$ is a Poisson process on the plane yields the desired independence.
\end{proof}

Now we show that $S_n$ is asymyptotically normal.
\begin{theorem}
\label{t2}
    Under assumption (A1) to (A5) we have that
    \begin{align}
        \sqrt{n}hS_n\overset{d}{\rightarrow} \mathcal{N}(0,\sigma^2) 
    \end{align}
    where 
    \begin{align}
       \sigma^2:= \int_{\R^2}k_y(x)^2f(x)dx\int_{\R^2}K_\infty(x)^2dx~.
    \end{align}
\end{theorem}
\begin{proof}
    We define the following random variables
    \begin{align}
        T_n&=\sqrt{n}h\sum_{j,k}V_{jk}\\
        T_{nR}&=\sqrt{n}h\sum_{j,k \leq Rc_n}V_{jk}
    \end{align}
    where $R$ is some arbitrary positive real number. We begin with noting that $\E[S_n]=\E[T_n]=\E[T_{nr}]=0$ and that (using \eqref{p4})
    \begin{align}
    \label{p3}
        \text{Var}(V_{jk})=\frac{1}{nh^4}\int_{I_{jk}}k^2_y(x)\int_{\R^2}\tilde K_\infty\Big(\frac{x-u}{h}\Big)^2dF(u)dx
    \end{align}
    Standard arguments then show that
    \begin{align} 
    \label{p6}
       \text{Var}(T_{nR}) \rightarrow \int_{\norm{x}\leq R}k_y(x)^2f(x)dx\int_{\R^2}K(x)^2dx:=\sigma^2_R
    \end{align}
    Similarly one obtains that
    \begin{align}
    \label{p5}
        \E[V_{jk}^4]\lesssim \frac{c_n^4h^4}{n^2h^8}
    \end{align}
    which, by the Lyapunov CLT (applicable because of \eqref{p3} and \eqref{p5}), yields that
    \begin{align}
    \label{p1}
        T_{nR}\overset{d}{\rightarrow}  \mathcal{N}(0,\sigma^2_R)~.
    \end{align}
    We also observe the following additional facts which follow from the dominated convergence theorem:
    \begin{align}
    \label{p2}
        \lim_n& \lim_R \  \p(|T_{n,r}-T_n|>\epsilon)=0\\
        &N(0,\sigma^2_R)\overset{d}{\rightarrow} N(0,\sigma^2)        
    \end{align}
    Combining \eqref{p1} and \eqref{p2} we obtain by Theorem 3.2 from \cite{billingsley1968} the desired conclusion.
\end{proof}

\subsection{Bootstrap}

We need the bootstrap versions of the auxiliary variables we defined in the proof of the asymptotic result. We will list them below for the readers convenience. Let $F_n^*$ denote the empirical distribution function of the (inaccessible) bootstrap sample $(X_1^*,Y_1^*),...,(X_n^*,Y_n^*)$ induced by $Z_1^*,...,Z_n^*$. 
\begin{align}
    f^*_n(x)&=\frac{1}{nh^2}\sum_{k=1}^nK_n((x-(X^*_k,Y^*_k))/h)\\
    k^*_{y,n}(x)&=\int_{\R^2}k(x,y)f_n^*(y)dy\\
    \tilde f^*&=h^{-2}(\tilde K_\infty\star F_n^*)\\
    f_n^{K,*}&=h^{-2}(K_\infty \star F_n^*)\\
    f_n^{po,*}&=h^{-1}(\tilde K_\infty \star F_n^*)\\
    F_n^{po,*}&=\frac{N}{n}F_N^*\\    
    V_{jk}^*&=\int_{I_{jk}}k_{y,n}(x)\int_{\R^2}\tilde K_\infty\Big(\frac{x-y}{h}\Big)d(F_n^{po,*}-F_n)(y)dx\\
    S_n^*&=\int_{\R^2}k_{y,n}(x)(f^{po,*}(x)-\E^*[f^{po,*}(x)])dx
\end{align}
We also define for any measurable set $A$ a shorthand for the conditional probability and expectation given $Z_1,...,Z_n$ as follows
\begin{align}
    \P^*(A)&=\P(A | Z_1,...,Z_n)\\
    \E^*[A]&=\E[A | Z_1,...,Z_n]~.
\end{align}
\textbf{Proof of Theorem \ref{t3}}\\
\begin{proof}
Follows by the same arguments as Theorem \ref{t1}, using Lemmas \ref{bl2} to \ref{bl6} and Theorem \ref{bt2} in place of Lemmas \ref{pl2} to \ref{pl6} and Theorem \ref{t2}.
\end{proof}

\begin{Lemma}
\label{bl2}
    Suppose that both assumptions (A1)-(A5) and (B1)-(B2) hold. Then there exists a set $\mathcal{A}$ with $\p(\mathcal{A})=1-o(1)$ on which
    \begin{align}
    \label{b1}
        \norm{f_n^*(x)-f_n(x)}_1 \lesssim h^{-\beta-1}n^{2/\kappa-1/2}\log(n)
    \end{align}
    holds with $\P^*$ probability $1-o(1)$.
\end{Lemma}
\begin{proof}    
    Following the arguments in the proof of Lemma \ref{pl2} verbatim we arrive at the inequality
    \begin{align}
        \norm{f_n^*(x)-f_n(x)}_1 \lesssim h^{-\beta-1}\norm{F_n^*-F_n}_1
    \end{align}
    Using that we have $\E[|X_i|^\kappa]$ and $\E[|Y_i|^\kappa]$ are both finite we obtain by the union bound that
    \begin{align}
        \max_{1 \leq i \leq n}\norm{(X_i,Y_i)}_\infty  =O_\p(n^{2/\kappa})
    \end{align}
    In particular we have that there exists a set  $\mathcal{A}$ with $\P(\mathcal{A})=1-o(1)$ on which
    \begin{align}
        \max_{1 \leq i \leq n}\norm{(X_i,Y_i)}_\infty  \lesssim n^{2/\kappa}\sqrt{\log(n)}
    \end{align}
    holds. This yields that
    \begin{align}
        \int_{\R^2}\sqrt{F_n(t)(1-F_n(t)}dt \lesssim n^{2/\kappa}\log(n)
    \end{align}
    holds on $\mathcal{A}$. In particular we then use conditional versions of the Markov inequality and the union bound to obtain that
    \begin{align}
        \norm{f_n^*(x)-f_n(x)}_1 \lesssim h^{-\beta-1}n^{2/\kappa-1/2}\log(n)
    \end{align}
    holds with $\P^*$-probability $1-o(1)$ on the set $\mathcal{A}$. We are done. 
\end{proof}

\begin{Lemma}
    \label{bl3}
    Suppose that both assumptions (A1)-(A5) and (B1)-(B2) hold. Then there exists a set $\mathcal{A}$ with $\p(\mathcal{A})=1-o(1)$ on which
    \begin{align}
        sup_{y \in \R^2}| k^*_{x,n}(y)-k_{x,n}(y)|&\lesssim h^{-\beta-1}n^{2/\kappa-1/2}\log(n)\\
        \Big|\int_{\R^2}( k^*_{x,n}(y)-k_{x,n}(y))(f_n^*(y)-f_n(y))dy\Big| & \lesssim h^{-2\beta-2}n^{4/\kappa-1}\log(n)
    \end{align}
    holds with $\p^*$ probability $1-o(1)$.
\end{Lemma}
\begin{proof}
     Can be carried over verbatim from the proof of Lemma \ref{pl3}, using Lemma \ref{bl2} in place of \ref{pl2}.
\end{proof}

\begin{Lemma}
    \label{bl5}
    We have that
    \begin{align}
        n^{1/2}h\norm{\tilde f^*-h^\beta f_n^*}_1 =o(1)
    \end{align}
    holds a.s. conditional on $Z_1,...,Z_n$.
\end{Lemma}
\begin{proof}
    Can be carried over verbatim from the proof of Lemma \ref{pl5}.
\end{proof}

\begin{Lemma}
    \label{bl6}
    Suppose that both assumptions (A1)-(A5) and (B1)-(B2) hold. We have that the random variables
    \begin{align}
     V_{jk}^*:= \int_{I_{jk}}k_y(x)\int_{\R^2}\tilde K_\infty\Big(\frac{x-y}{h}\Big)d(F^{po,*}_n-F_n)(y)dx, \quad j,k \in \mathbb{Z}
\end{align}
are independent, conditionally on $Z_1,...,Z_n$.
\end{Lemma}
\begin{proof}
    Can be carried over verbatim from the proof of Lemma \ref{pl6}.
\end{proof}

\begin{theorem}
    \label{bt2}
    Suppose that both assumptions (A1)-(A5) and (B1)-(B2) hold. We then have that
    \begin{align}
        \sqrt{n}hS_n^*\overset{d}{\rightarrow} \mathcal{N}(0,\sigma^2) 
    \end{align}
    conditionally on $Z_1,...,Z_n$ in probability.
\end{theorem}
\begin{proof}
    We follow the strategy of the proof of Theorem \ref{t2}. First we need to establish the analogues to equations \eqref{p6} and \eqref{p5}. To that end note that, due to $\tilde K_\infty$ being bounded on $\R^2$ and being 0 outside of $[-a_n,a_n]^2$, we have on a set $\mathcal{A}$ with $\p(\mathcal{A})=1-o(1)$ that
    \begin{align}
    \label{b2}
        \int_{\R^2}\tilde K_\infty\Big(\frac{x-u}{h}\Big)^2DF_n(u)&\lesssim F_n([x_1-a_nh,x_1+a_nh]\times [x_2-a_nh,x_2+a_nh]) \\
        &\lesssim a_n^2h^2+ \Big|(F_n-F)([x_1-a_nh,x_1+a_nh]\times [x_2-a_nh,x_2+a_nh])\Big|\\
        &\lesssim a_n^2h^2
    \end{align}
    holds. We hence obtain that on $\mathcal{A}$ it holds that
    \begin{align}
    \label{b3}
        \E[(V_jk^*)^4]\lesssim \frac{a_n^4h^4c_n^4}{n^2h^8}~.
    \end{align}
    By the same arguments as in the proof of Lemma \ref{bl2} we also obtain that
    \begin{align}
        \Big|\text{Var}(T_{nR})-\text{Var}^*(T_{nR}^*)\Big| \leq \frac{\log(n)}{\sqrt{n}h}
    \end{align}
    holds on $\mathcal{A}$. Taking subsequences we may assume that \eqref{b2} and \eqref{b3} hold almost surely. Taking $a_n=n^\delta$ for a sufficiently small $\delta$ then yields that we may apply the Lyapunov CLT to obtain that
    \begin{align}
        T_{nR}^* \overset{d}{\rightarrow }\mathcal{N}(0,\sigma_R^2)
    \end{align}
    holds a.s. conditionally on $Z_1,...,Z_n$. The conditional a.s. analogue to \eqref{p2} is also an obvious consequence of the dominated convergence theorem. This yields the desired statement along the subsequence we took. We can argue like this along a subsubsequence of any subsequence which yields the desired result by the metrizability of weak convergence.
\end{proof}

\bibliographystyle{apalike}
\bibliography{deconvolution_dependence}

\end{document}